%% file: item-prices.tex
\documentclass[11pt,letterpaper]{article}

\usepackage[letterpaper,margin=1in]{geometry}

\usepackage[noend]{algorithmic,algorithm}

\usepackage{amsmath}
\usepackage{amssymb}
\usepackage{mathtools}

\usepackage{dsfont}

\usepackage[dvipsnames]{xcolor}

\newcommand{\comment}[1]{}

\usepackage{amsthm}

\newenvironment{proofsketch}{\begin{proof}[Proof Sketch]}{\end{proof}}

\newcommand{\E}{\mathbb{E}}

\newcommand{\MOPH}{\mathcal{MPH}}

\newcommand{\XOS}{\mathcal{XOS}}
\newcommand{\nonnegR}{\mathbb{R}^+}

\newcommand{\alg}{\mathcal{A}}

\newtheorem{theorem}{Theorem}[section]
\newtheorem{lemma}[theorem]{Lemma}

\newtheorem{corollary}[theorem]{Corollary}

\newtheorem{definition}{Definition}[section]

\floatstyle{ruled}
\newfloat{algorithm}{tbp}{loa}
\floatname{algorithm}{Algorithm}

\newcommand{\reals}{\mathbb{R}}
\newcommand{\Omit}[1]{}

\begin{document}

\input{notations}

\title{Combinatorial Auctions via Posted Prices}

\author{
Michal Feldman\thanks{Tel-Aviv University; {\tt mfeldman@tau.ac.il}}
\and
Nick Gravin\thanks{Microsoft Research; {\tt ngravin@microsoft.com}}
\and
Brendan Lucier\thanks{Microsoft Research; {\tt brlucier@microsoft.com}}
}
\date{}

\maketitle

\begin{abstract}
\input{abstract}

\end{abstract}


\section{Introduction}
\input{introduction}

\section{Preliminaries}
\label{sec:model}

\input{model}

\section{Posted Prices for XOS Valuations}
\label{sec:xos}
\input{xos}

\section{Posted Prices for General Valuations}
\label{sec:mph}
\input{MPH}

\section{Discussion and Open Problems}
\label{sec:conclusions}
\input{conclusions}

\section*{Acknowledgments}
We are grateful to Shaddin Dughmi, David Kempe and Noam Nisan for helpful discussions.
Michal Feldman is partially supported by the European Research Council under the European Union's Seventh Framework Programme (FP7/2007-2013) / ERC grant agreement number 337122.

\bibliographystyle{plain}
\bibliography{item-prices}

\medskip

\appendix
\section*{APPENDIX}
\setcounter{section}{0}


%
%


\section{Details of the proof of Theorem \ref{th:XOS_Bayesian_Comp}}
\label{app:price-sample}
\input{price-sample-appendix}

\section{Proof of Theorem~\ref{th:mohk_Bayesian_Comp}: $\MOPH$ valuations}
\label{app:mph}

\input{MPH-app}

\section{Lower bound for $\MOPH$}
\label{app:mph-lower}
\input{mph-lower}

\end{document}

%% file: notations.tex

\newcommand{\AutoAdjust}[3]{\mathchoice{ \left #1 #2  \right #3}{#1 #2 #3}{#1 #2 #3}{#1 #2 #3} }
\newcommand{\Xcomment}[1]{{}}
\newcommand{\eval}[1]{\left.#1\vphantom{\big|}\right|}
\newcommand{\inteval}[1]{\Big[#1\Big]}
\newcommand{\InParentheses}[1]{\AutoAdjust{(}{#1}{)}}
\newcommand{\InBrackets}[1]{\AutoAdjust{[}{#1}{]}}
\newcommand{\Ex}[2][]{\operatorname{\mathbf E}_{#1}\InBrackets{#2\vphantom{E_{F}}}}
\newcommand{\Exlong}[2][]{\operatornamewithlimits{\mathbf E}\limits_{#1}\InBrackets{#2\vphantom{\operatornamewithlimits{\mathbf E}\limits_{#1}}}}
\newcommand{\Prx}[2][]{\operatorname{\mathbf{Pr}}_{#1}\InBrackets{#2}}
\newcommand{\Prlong}[2][]{\operatornamewithlimits{\mathbf Pr}\limits_{#1}\InBrackets{#2\vphantom{\operatornamewithlimits{\mathbf Pr}\limits_{#1}}}}
\def\prob{\Prx}
\def\expect{\Ex}
\newcommand{\super}[1]{^{(#1)}}
\newcommand{\dd}{\mathrm{d}}  
\newcommand{\given}{\;\mid\;}

\newcommand{\Ell}{\ensuremath{\mathcal{L}}}

\newcommand{\be}{\begin{equation}}
\newcommand{\ee}{\end{equation}}
\newcommand{\argmin}{\mathop{\rm argmin}}
\newcommand{\argmax}{\mathop{\rm argmax}}
\newcommand{\vr}[1]{{\mathbf{#1}}}
\newcommand{\bydef}{\stackrel{\bigtriangleup}{=}}
\newcommand{\eps}{\varepsilon}
\newcommand{\mm}[1]{\mathrm{#1}}
\newcommand{\mc}[1]{\mathcal{#1}}
\newcommand{\mb}[1]{\mathbf{#1}}
\newcommand{\vect}[1]{\ensuremath{\mathbf{#1}}}
\newcommand{\R}{\mathbb{R}}

\def \EE   {{\mathbb E}}
\def \OPT {\mathcal{OPT}}
\def \vf  {\textrm{vf}}
\def \dvf {\varphi}
\def \utility {u}
\def \reals {{\mathbb R}}

\newcommand{\poly} {\text{\textnormal{POLY}}}
\newcommand{\MPH}{\text{MPH}}
\newcommand{\MPHk}[1][k]{\MPH-$#1$ }

\newcommand{\Idr}[1]{\mathds{1}\InBrackets{#1\vphantom{\sum}}}


\newcommand{\dist}{\mathcal{F}}
\newcommand{\disti}[1][i]{{\mathcal{F}_{#1}}}
\newcommand{\distsmi}[1][i]{\dists_{\text{-}#1}}
\newcommand{\dists}{\vect{\dist}}
\newcommand{\distw}{\widetilde{\mathcal{F}}}

\newcommand{\valdist}{\mathcal{F}}
\newcommand{\valdists}{\vect{\valdist}}
\newcommand{\valdisti}[1][i]{{\valdist_{#1}}}
\newcommand{\valdistsmi}[1][i]{\valdists_{\text{-}#1}}

\newcommand{\dens}{f}
\newcommand{\denss}{\vect{ \dens}}
\newcommand{\densi}[1][i]{{\dens_{#1}}}

\newcommand{\agents}{N}
\newcommand{\nagent}{n}
\newcommand{\goods}{M}
\newcommand{\nitem}{m}
\newcommand{\auction}{A}

\newcommand{\weight}{w}
\newcommand{\partition}{\Gamma}
\newcommand{\partitioni}{\Gamma_i}

\newcommand{\CWE}[0]{\textsf{CWE}}
\newcommand{\ef}[0]{envy free}
\newcommand{\EF}[0]{EF}
\newcommand{\CWEWOMC}[0]{CWE without market clearance}
\newcommand{\SW}[0]{\textsf{SW}}
\newcommand{\Revenue}[0]{\textsf{Rev}}

\newcommand{\bid}{b}
\newcommand{\bids}{\vect{\bid}}
\newcommand{\bidsmi}[1][i]{\bids_{\text{-}#1}}
\newcommand{\bidi}[1][i]{{\bid_{#1}}}

\newcommand{\val}{v}
\newcommand{\vals}{\vect{\val}}
\newcommand{\valsmi}[1][i]{\vals_{\text{-}#1}}
\newcommand{\vali}[1][i]{{\val_{#1}}}
\newcommand{\valith}[1][i]{{\val_{(#1)}}}

\newcommand{\wal}{\widetilde{v}}
\newcommand{\wals}{\vect{\wal}}
\newcommand{\walsmi}[1][i]{\wals_{\text{-}#1}}
\newcommand{\wali}[1][i]{{\wal_{#1}}}

\newcommand{\util}{u}
\newcommand{\utils}{\vect{\util}}
\newcommand{\utili}[1][i]{\util_{#1}}
\newcommand{\utilsmi}[1]{\utils_{\text{-}#1}}

\newcommand{\price}{p}
\newcommand{\prices}{\vect{\price}}
\newcommand{\pricei}[1][i]{{\price_{#1}}}

\newcommand{\pricep}{p^{\prime}}
\newcommand{\pricesp}{\vect{\price}^{\prime}}
\newcommand{\priceip}[1][i]{{\price^{\prime}_{#1}}}

\newcommand{\qprice}{q}
\newcommand{\qprices}{\vect{\qprice}}
\newcommand{\qpricei}[1][i]{{\qprice_{#1}}}

\newcommand{\payment}{\rho}
\newcommand{\payments}{\vect{\payment}}
\newcommand{\paymenti}[1][i]{{\payment_{#1}}}

\newcommand{\type}{t}
\newcommand{\types}{\vect{\type}}
\newcommand{\typei}[1][i]{{\type_{#1}}}
\newcommand{\typesmi}[1][i]{\types_{\text{-}#1}}

\newcommand{\alloc}{X}
\newcommand{\allocs}{\vect{\alloc}}
\newcommand{\allocsmi}[1][i]{\allocs_{\text{-}#1}}
\newcommand{\alloci}[1][i]{{\alloc_{#1}}}

\newcommand{\opt}{\text{OPT}}
\newcommand{\opts}{\vect{ \opt}}
\newcommand{\opti}[1][i]{{\opt_{#1}}}
\newcommand{\optsmi}[1][i]{\opts_{\text{-}#1}}

\newcommand{\talloc}{Y}
\newcommand{\tallocs}{\vect{\talloc}}
\newcommand{\talloci}[1][i]{{\talloc_{#1}}}
\newcommand{\tallocsmi}[1][i]{\tallocs_{\text{-}#1}}

\newcommand{\rank}{r}
\newcommand{\ranks}{\vect{\rank}}
\newcommand{\ranki}[1][i]{{\rank_{#1}}}
\newcommand{\ranksmi}[1][i]{\ranks_{\text{-}#1}}

\newcommand{\crit}{\theta}
\newcommand{\crits}{\vect{\crit}}
\newcommand{\criti}[1][i]{{\crit_{#1}}}
\newcommand{\critsmi}[1][i]{\crits_{\text{-}#1}}

\newcommand{\decl}{d}
\newcommand{\decls}{\vect{\decl}}
\newcommand{\decli}[1][i]{{\decl_{#1}}}
\newcommand{\declsmi}[1][i]{\decls_{\text{-}#1}}

\newcommand{\demand}{D}
\newcommand{\demands}{\vect{\demand}}
\newcommand{\demandi}[1][i]{{\demand_{#1}}}
\newcommand{\demandsmi}[1][i]{\demands_{\text{-}#1}}

\newcommand{\sold}{{\texttt{SOLD}}}
\newcommand{\soldi}[1][i]{\sold_{#1}}

\newcommand{\CWEalg}{{\sc CWE} algorithm}

\newcommand{\Pool}{\text{Pool}}
\newcommand{\Pop}{\textbf{Pop}}
\newcommand{\Push}{\textbf{Push}}
\newcommand{\Bundle}{\textbf{Bundle}}
\newcommand{\ResolveConflict}{\textbf{ResolveConflict}}


\newcommand{\Reject}{\text{Reject}}
\newcommand{\AllocateDemand}{\textbf{AllocateDemand}}
\newcommand{\RaisePrices}{\textbf{RaisePrices}}

%% file: abstract.tex
%

We study anonymous posted price mechanisms for combinatorial auctions in a Bayesian framework.
In a posted price mechanism, item prices are posted, then the consumers approach the seller sequentially in an arbitrary order, each purchasing her favorite bundle from among the unsold items at the posted prices.
These mechanisms are simple, transparent and trivially dominant strategy incentive compatible (DSIC).

We show that when agent preferences are fractionally subadditive (which includes all submodular functions), there always exist prices that, in expectation, obtain at least half of the optimal welfare. Our result is constructive: given black-box access to a combinatorial auction algorithm $A$, sample access to the prior distribution, and appropriate query access to the sampled valuations, one can compute, in polytime, prices that guarantee at least half of the expected welfare of $A$. As a corollary, we obtain the first polytime (in $n$ and $m$) constant-factor DSIC mechanism for Bayesian submodular combinatorial auctions, given access to demand query oracles. Our results also extend to valuations with complements, where the approximation factor degrades linearly with the level of complementarity.

%% file: introduction.tex
The canonical problem in market design is to efficiently allocate a set of $m$ resources among a set of $n$ self-interested agents.  Such allocation problems range in scope from the trade of a single item between a seller and a buyer, to combinatorial auctions in which many heterogeneous goods are to be divided among multiple participants with complex and idiosyncratic preferences.  Scenarios of the latter type have attracted significant recent attention from the computer science community, due to algorithmic challenges presented by the underlying allocation problem.  For example, the efficient allocation of cloud resources involves the scheduling of computing tasks, and the allocation of wireless spectrum involves finding large independent sets in graphs that represent interference constraints.  The primary challenge in AGT is to marry algorithmic solutions to such problems with the economic principles that underpin market design.


As an example, take the problem of designing an incentive compatible mechanism for combinatorial auctions with submodular bidders.  The underlying optimization problem is NP-hard, and simple greedy methods achieve a constant approximation.  On the other hand, whether there is a truthful polytime constant-factor approximation has remained a vexing and major open question for over a decade\footnote{For some models of valuation access, such as the value query model, it is in fact known that no sub-polynomial worst-case approximation is possible~\cite{DughmiV11}.}.
One might therefore wonder whether it is even feasible to implement a combinatorial auction mechanism that is both computationally tractable and economically appealing.

When designing an economically viable mechanism, it is desirable to not only respect incentives and computational constraints, but also to resolve allocation decisions in a straightforward and transparent fashion.  For example, a simple and natural approach is to resolve a market using {\em posted prices}.
One might imagine an implementation in which items are assigned {\em anonymous} prices (i.e., all agents face the exact same prices), then agents arrive to the market and each consumes his most-desired bundle under the given prices.
Such a methodology exhibits many desirable properties: it is simple, decentralized, and easy to implement.
It is also transparent in the sense that a buyer does not need to understand the price-setting method in order to understand how to participate --- he simply behaves as a price-taker and consumes his preferred bundle.
Therefore, it is also trivially incentive compatible.

Posted price mechanisms are highly applicable when markets are large and the aggregate demands of buyers can be accurately predicted.
For instance, if the buyers' valuations are public knowledge and satisfy the gross substitutes condition, then there always exist prices that efficiently clear the market \cite{Gul1999,Kelso1982}.
Similar results hold in large markets for arbitrary valuations \cite{Azevedo2013}.
However, if valuations are not fully known to the seller and are unpredictable, as in the submodular combinatorial auction problem described above, then it is unclear how reasonable prices can be set.
This motivates the approach of implementing an optimization algorithm as an {\em auction}, where buyers submit competing bids that are treated as input.
As Milgrom writes, ``When goods are not standardized or when the market clearing prices are highly unstable, posted prices work poorly, and auctions are usually preferred" \cite{Milgrom89}.

We are left with a dichotomy.
Posted prices form a simple and natural market instrument, but they generate efficient outcomes only under special circumstances and full information over the buyers' preferences.
General auctions are more widely applicable, but can be significantly more complex to execute and participate in; moreover, we still do not have reasonable auction designs for many allocation problems of interest.
A natural question arises:

\begin{quote}
To what extent can \emph{approximately} efficient outcomes be implemented using anonymous posted prices in settings of incomplete information?
\end{quote}

In this work we study the power of anonymous posted prices in Bayesian settings, where the designer knows the distribution over the agents' valuations but not their realizations.
The Bayesian setting imposes additional challenges over the full information setting typical in the market equilibrium literature.
In particular, readers familiar with this literature will note that, in this probabilistic setting, there do not exist prices that are guaranteed to satisfy all agent demands simultaneously.
To mitigate this problem, we consider posted price mechanisms that admit the agents sequentially, in an arbitrary order, to select their most preferred bundle from the remaining items.
%
%

We devise polytime posted price mechanisms, for several classes of valuations, that achieve nearly optimal social welfare given appropriate access to the distribution of the agents' valuations.
Specifically, for XOS valuations (a strict superset of submodular valuations) we devise a mechanism that obtains a constant fraction of the optimal social welfare.
Notably, this implies that we can obtain an $O(1)$-approximate dominant strategy incentive compatible mechanism for Bayesian XOS valuations, whose running time is polynomial in $n$ and $m$ given access to demand queries.
Prior to this work, the best-known polytime mechanisms (posted-price or otherwise) were either polylogarithmic approximations \cite{Dobzinski07} or had runtimes that were polynomial in the support size of an agent's valuation distribution, which could be exponential in $n$ and $m$ \cite{BH-11,HartlineKM11}.
In addition, for general monotone valuations, we devise a mechanism that obtains an approximation factor that degrades gracefully with the level of complementarity of the functions, as captured by the maximum-over-hypergraph (MPH) hierarchy, recently introduced in \cite{Feige14}.
A more detailed description of our results, along with a comparison to the related literature, appears below.

\subsection{Our model}

Our setting consists of a set $\goods$ of $\nitem$ indivisible objects and a set of $n$ buyers.
Each buyer has a valuation function $\vali(\cdot) : 2^\goods \to \reals_{\geq 0}$ that indicates his value for
every set of objects.  We assume valuations are monotone non-decreasing, normalized so that $\vali(\emptyset) = 0$, and scaled to lie in $[0,1]$.
The profile of buyer valuations is denoted by $\vals=(\val_1,\dotsc,\val_n)$.

An {\em allocation} of $M$ is a vector of sets $\allocs = (\alloc_1, \dotsc, \alloc_n)$, where $\alloci$ denotes the bundle assigned to buyer $i$, for every buyer $i \in [n]$, and $\alloc_i \cap \alloc_k = \emptyset$ for every $i\neq k$ (note that it is not required that all items are allocated).
The social welfare of an allocation $\allocs$ is $\SW(\allocs)=\sum_{i=1}^{n}\vali(\alloci)$, and the optimal welfare is denoted by OPT.

The utility of buyer $i$ being allocated bundle $\alloc_i$ under prices $\prices$ is $\utili(\alloci, \prices) = \vali(\alloci) - \sum_{j \in \alloci}\price_j.$ Given prices $\prices=(\price_1, \dotsc, \price_{m})$, the {\em demand correspondence} $\demandi(M, \prices)$ of buyer $i$ contains the sets of objects that maximize buyer $i$'s utility. 


We consider a {\em Bayesian} setting, where the bidders' valuations are drawn independently from distributions $\valdist_1, \ldots, \valdist_n$.  Write $\valdists = \valdist_1 \times \cdots \times \valdist_n$, so that $\vals$ is drawn from $\valdists$.  We think of $\valdists$ as being public knowledge, whereas the realization $\vali$ is known only to agent $i$.
In the Bayesian framework, an allocation $\allocs$ is said to be an $\alpha$-approximation (for social welfare) if
$$\Ex[\vals \sim \dists]{\SW(\allocs(\vals))} \ge (1/\alpha) \cdot \Ex[\vals \sim \dists]{\opt(\vals)}.$$

\subsection{Anonymous Posted Price Mechanisms}

The posted price mechanisms considered in this paper proceed in the following steps:
\begin{enumerate}
\setlength{\itemsep}{1pt}
\setlength{\parskip}{0pt}
\setlength{\parsep}{0pt}
\item (pricing phase) a price vector $\prices$ is determined, based on $\dists$;
\item (arrival order phase) an arrival order $\pi$ is determined;
\item (value realization phase) A value $\vali \sim \disti$ is realized for every buyer $i$, known only to buyer $i$;
\item (consumption phase) The buyers arrive to the market according to the order $\pi$, and each buyer receives his most desired bundle among all remaining items. That is, for every buyer $i$, $\alloci \in \demandi(\prices, M \setminus \bigcup_{j <_{\pi} i} \alloci[j])$. Buyer $i$ pays the sum of item prices in his bundle $\alloci$.
\end{enumerate}

We refer to these as anonymous posted price mechanisms, though sometimes we will leave out ``anonymous'' for brevity.
We note that the mechanism can be described in two equivalent ways, namely as an indirect or a direct mechanism.
In both versions, the mechanism sets the prices and determines\footnote{All of our results hold under arbitrary arrival orders, so one can also think of an adversary as choosing the order.} the arrival order.
The difference is that in the indirect implementation, the agents arrive and purchase their desired items in the determined order, whereas in the direct revelation implementation, the buyers report their valuations, and the mechanism simulates the consumption phase on their behalf.
This latter simulation requires that the mechanism have access to demand oracles\footnote{
While answering a demand query might be NP-hard in some cases, in the context of combinatorial auctions it is natural to expect the agent to be able to answer a demand query.  Otherwise it would be unreasonable to expect any mechanism to satisfy the agents' demands. Furthermore, in some cases there are polytime algorithms for answering demand queries. For example, for gross substitutes valuations demand queries can be implemented with a polynomial number of value queries~\cite{Renato13}.
}.

It is easy to verify that any posted price mechanism that adheres to this structure is trivially dominant strategy incentive compatible (DSIC).
Moreover, it is clearly {\em weakly group strategyproof}, meaning that no coalition of agents can deviate in a way that strictly benefits each one of them.



All of our results hold regardless of the order selected in the arrival order phase, so we can think of the arrival order as being chosen adversarially.  In the direct revelation implementation it would be natural to imbue the mechanism with the power to choose the ordering, but we view obliviousness to the order as increasing the robustness of the mechanism.

Beyond the strong incentive characteristics of posted price mechanisms, they are appealing in their simplicity.
In particular, they proceed by setting a single price vector and using it for all agents.


\subsection{Our results}

We establish welfare guarantees for XOS (i.e., fractionally subadditive) valuations (a strict superset of submodular valuations), and for MPH-$k$ valuations --- a hierarchy that spans all monotone valuations, parameterized by the complementarity level $k$.
%
%
%
Our main contribution is the following:

\vspace{0.15in}\noindent \textbf{Theorem:}
{\bf [2-approximation for XOS]}
Given black-box access to an algorithm $\alg$ for XOS valuations, sample access to $\XOS$ distributions $\dists$, an $\XOS$ query oracle for the valuations in the support of $\dists$, and a demand oracle for the valuations, there exists a posted price mechanism that, for every XOS valuation profile $\vals$, and every $\eps$, returns an outcome that gives expected social welfare of at least $\frac{1}{2}\E_{\vals \sim \dists}[\alg(\vals)]-\eps$, and runs in time $\poly(n,m,1/\eps)$.

\vspace{0.2in}We note that the factor of 2 is tight (even for a single item), as established by the illustrating example in the end of this section.
We also note that the demand oracle requirement is needed only for simulating agent behavior in the consumption phase.
This theorem implies a set of results for XOS valuations and special cases thereof, including submodular and gross substitutes valuations.
Before presenting the corollaries, we wish to emphasize the strengths of this result in light of the previous literature:
\begin{itemize}
\setlength{\itemsep}{1pt}
\setlength{\parskip}{0pt}
\setlength{\parsep}{0pt}
\item {\bf (Truly) polytime:} Many previous mechanisms have (essentially) pseudo-polynomial runtime, in the sense that they are polynomial in the size of the agents' type space, which may plausibly be exponential in $m$ and $n$.
The runtime of our mechanisms is polynomial in $m$ and $n$, independent of the type space sizes.
\item {\bf DSIC:} Many existing mechanisms for settings of incomplete information exhibit the weaker notion of Bayesian incentive compatibility
(BIC). Our mechanism exhibits the stronger notion of DSIC, and is also weakly group strategyproof.
\end{itemize}


The heart of the construction of our mechanism lies in the appropriate price assignment; this is the part where oracle access is needed (except for the consumption phase that requires access to demand oracles).
Together with known algorithmic results, our theorem implies the following polytime approximation results for XOS, submodular and gross substitutes valuations.

\vspace{0.15in}\noindent \textbf{Theorem:}
{\bf [Computational results]}
We devise polytime (in $n$ and $m$) DSIC mechanisms with the following guarantees:
\begin{itemize}
\setlength{\itemsep}{1pt}
\setlength{\parskip}{0pt}
\setlength{\parsep}{0pt}
\item {\bf XOS:} A $\frac{e}{2(e-1)}$-approximation, 
given sample access to the distribution, value and demand oracles, and XOS query oracles.
\item {\bf Submodular:} A $\frac{e}{2(e-1)}$-approximation, 
given sample access to the distribution, value oracles, and demand oracles.
Demand oracles are required only in the consumption phase.
\item {\bf Gross substitutes (GS):} A $\frac{1}{2}$-approximation, 
given only sample access to the distribution and value oracles.
\end{itemize}

So far we have only considered complement-free valuations.
Our second main result concerns valuation functions that exhibit complementarities.

\vspace{0.15in}\noindent \textbf{Theorem:}
{\bf [$k$-approximation for $\MPHk$]}
Given a black-box access to an algorithm $\alg$ for $\MPHk$ valuations, a sample access to $\MPHk$ distributions $\dists$, an $\MPHk$ query oracle for the valuations in the support of $\dists$, and a demand oracle for the valuations, there exists a posted price mechanism that, for every $\MPHk$ valuation profile $\vals$, and every $\eps$, returns an outcome that gives expected social welfare of at least $\frac{1}{k}\E_{\vals \sim \dists}[\alg(\vals)]-\eps$, and runs in time $\poly(n,m^k,1/\eps)$.

\vspace{0.15in}By the fact that the class of $\MPHk[m]$ valuations (i.e., $k=m$) is equivalent to the class of all monotone functions, the last theorem implies the existence of posted prices that gives expected social welfare of at least $\frac{1}{\Omega(m)}$ fraction of $\opt$. We also show that this bound is tight.

\paragraph{An Illustrating Example. }
To give some insight into our results, consider the case of a single item and $n$ bidders with values drawn i.i.d. from some distribution $F$.  In this setting, the Vickrey auction generates the efficient outcome.  How well can one approximate the efficient outcome by setting a single price $p$ (that depends only on $F$ and $n$) and allocating to a random bidder\footnote{Since agents are iid, this is equivalent to an adversarial order that doesn't depend on the value realizations.} with value greater than $p$, if any exists?

It is known that one cannot hope to achieve better than half of the optimal expected welfare~\cite{HK92}.  For instance, suppose $F$ is such that each agent has (large) value $X$ with probability $q = 1 - (1-1/X)^{1/n}$ and value $1$ otherwise.  Then the probability that \emph{any} agent has value $X$ is $1/X$, and hence the expected optimal welfare is $1 \cdot (1-1/X) + X \cdot (1/X)$, which approaches $2$ as $X$ grows large.  On the other hand, no posted price obtains welfare greater than $1$: if $p > 1$, it generates welfare $X$ with probability $1/X$; whereas if $p \leq 1$, an arbitrary agent will buy and the expected welfare is $1 + Xq = 1+O(\frac{1}{n})$.  A posted price therefore cannot extract more than half of $\opt$.

However, there is a simple pricing scheme for a single item that yields half of the optimal social welfare; such methods are known from the prophet inequality literature~\cite{KW12}.  Specifically, set price $\price$ equal to half of the expected highest value.  To see why this works, write $\pi$ for the probability that the item is sold, and write $i^*$ for the agent with the highest value for the item (note $i^*$ is a random variable). The expected revenue from the auction is precisely $\price \cdot \pi = \frac{1}{2}\E[v_{i^*}] \cdot \pi$. On the other hand, the probability that nobody buys the item ahead of buyer $i^*$ is at least $1-\pi$. The expected surplus (value minus payments) of buyer $i^*$, conditioned on item being still available, is at least $\E[v_{i^*}] - \price = \frac{1}{2}\E[v_{i^*}]$. Putting this together, we have that the expected welfare of the auction, which equals the expected revenue plus the expected buyer surplus, is at least $\frac{1}{2}\E[v_{i^*}] \cdot \pi +\frac{1}{2}\E[v_{i^*}] \cdot (1 - \pi) = \frac{1}{2}\E[v_{i^*}]$, as claimed.

Our main result shows that the reasoning given in the single-item example can be extended to markets with multiple heterogeneous items for sale and asymmetric buyers, as long as buyer preferences lie in the class of XOS valuations.

\subsection{Related work}

\input{related-work}

%% file: related-work.tex

Our work is part of the recent body of literature on simple, non-optimal mechanisms \cite{HartlineR09, DhangwatnotaiRY10}.
The design and performance of simple mechanisms is an active field of research, motivated by the observation that in practice, designers are often willing to trade truthfulness or optimality for simplicity.
Canonical examples include the generalized second price (GSP) auctions for online advertising \cite{Edelman05,Varian07}, and the ascending price auction for electromagnetic spectrum allocation \cite{Milgrom98}.

A particularly relevant example of a simple mechanism is the {\em simultaneous item} auction, in which agents make simultaneous, separate bids on multiple items.  
Such auctions are not truthful, but achieve, at equilibrium, a constant approximation to the optimal welfare when agents have complement-free valuations
\cite{ Bhawalkar12, CKS08, FFGL13, HKMN11}.  A conclusion from this is that by restricting bidders to bid only on individual items, rather than on packages as in the VCG auction, one loses only a constant factor from the optimal welfare.  Our results show that it is possible to go one step further, in terms of simplicity.
Indeed, posted price mechanisms not only handle items separately but also forego competition entirely and simply publish prices on individual items, and simultaneously exhibit strong incentive compatibility.

There is a long line of research studying the performance of posted price mechanisms under the objective of maximizing revenue.  When there is only a single item for sale, posted prices obtain $78\%$ of the optimal revenue in large markets \cite{BlumrosenH08}.  When agents have unit-demand preferences, a form of posted-price mechanism extracts a constant fraction of the optimal revenue \cite{Chawla07, Chawla10, ChawlaMS10}.  When there is only a single item for sale, this constant factor persists even when the distributions are unknown, as long as they are MHR \cite{Babaioff11}.
Notably, these works all apply a more relaxed notion of a posted-price mechanism, in which one can set different prices for each customer.  In contrast, the mechanisms we consider use non-discriminatory pricing.

Our work relates to the design of truthful submodular combinatorial auctions.  Given access to demand queries, a randomized truthful $O(\log m \log\log m)$ worst-case approximation exists \cite{Dobzinski07}.  It is a major open question whether there is a truthful constant-factor mechanism, using demand queries.  It is known that no sub-polynomial factor is possible under the value query model \cite{DughmiV11}, or for succinctly-described valuations \cite{DobzinskiV12}.
We show that in Bayesian settings, where the performance of the mechanism is evaluated based on its expected social welfare given an input distribution, one can indeed design a truthful constant-factor submodular combinatorial auction.

Ours is not the first work to turn to the Bayesian setting for combinatorial auction design. Hartline and Lucier~\cite{HartlineL10}, Hartline et al.~\cite{HartlineKM11} and Bei and Huang~\cite{BH-11} provide black-box reductions that convert an arbitrary welfare-maximization algorithm into an (approximately) Bayesian incentive compatible (BIC) mechanism, without loss of welfare.  In particular, one can apply the latter two to constant-factor approximations for submodular CAs to obtain BIC constant-factor mechanisms.  However, these mechanisms require time polynomial in the support size of an agent's valuation distribution, which can be exponential in $n$ and $m$.  In contrast, our mechanism runs in time polynomial in $n$ and $m$, regardless of the valuation distributions' support sizes.  Alaei~\cite{Saeed11} presents a general method for designing DSIC combinatorial auction mechanisms in Bayesian settings, using an algorithm for a related single-agent optimization problem, but does not consider submodular CAs\footnote{While Alaei~\cite{Saeed11} does not explicitly discuss submodular CAs, our understanding is that one could use his methodology together with an algorithm for (single-agent) submodular function maximization to construct a constant-factor DSIC mechanism for the submodular CA problem.  We note that such a mechanism would not fall within the posted-price paradigm.}.

%% file: model.tex
\paragraph{Valuation Classes} We study both complement-free valuations, and valuations that exhibit complementarities.
%
%
There is a standard hierarchy of complement-free valuations (see \cite{Lehmann2001}): additive $\subset$ gross substitutes $\subset$ submodular $\subset$ XOS $\subset$ subadditive.
\begin{description}
\item[\emph{additive}] $\val(S) = \sum_{j \in S}\val(\{j\})$ for all $S \subset \goods$. 
\item[\emph{submodular}] for every $S \subseteq T \subseteq M$ and $j \in M$, $\val(j|T) \leq \val(j|S)$, where $\val(j|S) \coloneqq \val(S \cup \{j\}) - \val(S)$.
\item[\emph{XOS}] there exists a collection of additive functions
$A_1(\cdot),\ldots ,A_k(\cdot)$ 
such that for every set $S \subseteq \goods$,  $\val(S) = \max_{1\le i\le k}A_i(S)$.
\item[\emph{subadditive}] 
for any
subsets $S_1,S_2\subset \goods$, $\val(S_1)+\val(S_2)\ge\val(S_1\cup S_2)$.
\end{description}

To study valuations with complements, we consider the hierarchy {\em maximum over positive hypergraphs} {\bf (MPH)}, introduced recently by \cite{Feige14}.
This hierarchy is general enough to encapsulate all monotone valuation functions, and its level captures the degree of complementarity.
We defer a formal description to Section \ref{sec:mph}.

\paragraph{Computational model}

An algorithm for the combinatorial auction problem receives as input a valuation profile $\vals$, and returns an allocation profile.  We write $\alg$ for an algorithm, and $\alg(\vals)$ for the allocation returned.
%
%
As any explicit description of $\vali: 2^\goods\to \R_{\ge 0}$ would have size exponential in $m$, it is usually assumed that there is an oracle access to
$\vali$. We consider the following oracles:

\begin{itemize}
\setlength{\itemsep}{1pt}
\setlength{\parskip}{0pt}
\setlength{\parsep}{0pt}
\item {\bf Value oracle} takes as input a set $T$, and returns $\vali(T)$;
\item {\bf Demand oracle} takes as an input a price vector $\prices$, and returns a set from demand correspondence $\demandi(M, \prices)$, breaking ties arbitrarily but consistently; 
\item {\bf XOS oracle} (only for XOS function $\vali$) takes as input a set $T$, and returns the corresponding additive representative function for the set $T$, i.e., an additive function $A_i(\cdot)$ such that (i) $\vali(S)\ge A_i(S)$ for any $S\subset[m]$, and (ii) $\vali(T)=A_i(T)$;
\end{itemize}

While value oracle is the least computationally demanding for the buyers and the seller, the demand oracle captures the most basic decision problem a buyer faces in a market with item prices. Since the primary focus of this paper is on the pricing mechanisms, we assume throughout the paper an access to demand and value oracles for granted. We note that XOS 
oracles are less commonly used in the literature. However, for some classes of valuations XOS oracle can be implemented via polynomially many queries to value oracle, e.g., for any submodular function.

%% file: xos.tex
The main theorem in this section is the following.
\begin{theorem}
\label{th:XOS_Bayesian_Comp}
Let distribution $\dists$ over $\XOS$ valuation profiles be given via a sample access to $\dists$.
Suppose that for every $\vals\sim\dists$ we have
\begin{enumerate}
\item black-box access to a welfare maximization algorithm $\alg$ for combinatorial auctions,
\item an $\XOS$ query oracle (for valuations sampled from $\dists$).
\end{enumerate}
Then, for any $\eps>0$, we can compute item prices in $\poly(m,n,1/\eps)$ time such that, for any buyer arrival order, the expected welfare of the posted price mechanism is at least $\frac{1}{2}\E_{\vals \sim \dists}[\SW(\alg(\vals))] - \eps$.
\end{theorem}



\paragraph{Implications.} Before proving Theorem \ref{th:XOS_Bayesian_Comp}, let us discuss some implications.  First, note that using an $\alpha$-approximation algorithm for $\alg$ in Theorem \ref{th:XOS_Bayesian_Comp} results in a posted price mechanism with approximation factor $\alpha/2$, minus an additive error term that can be made as small as desired.
Recall that implementing the consumption phase in a direct revelation mechanism does require access to demand queries; note that Theorem \ref{th:XOS_Bayesian_Comp} and its corollaries below refer specifically to the pricing phase.

If we assume access to demand oracles, then we can use the polytime algorithm of Feige~\cite{Feige09} with approximation factor $1-1/e$ as a black box.  Theorem \ref{th:XOS_Bayesian_Comp} then implies the existence of a DSIC mechanism with expected social welfare at least $\frac{e}{2(e-1)}\opt - \eps$ and runtime $\poly(m,n,1/\eps)$.

For submodular valuations, one could instead use the algorithm by Vondrak~\cite{Vondrak08} with tight approximation factor $\alpha=1-1/e$ that utilizes only value queries.  Since $\XOS$ queries can be simulated by value queries for submodular valuations~\cite{Blumrosen2009}, we obtain the following corollary:
\begin{corollary}
\label{cor:submodular-prices}
Given sample access to submodular distributions $\dists$ and value oracle access to each valuation in the support of $\dists$,
for every $\eps>0$, one can compute item prices in time $\poly(m,n,1/\eps)$, such that, for any buyer arrival order, the expected welfare of the posted price mechanism is at least $\frac{e}{2(e-1)}\E_{\vals \sim \dists}[\SW(\alg(\vals))] - \eps$.
\end{corollary}

As before, one can implement the mechanism from Corollary~\ref{cor:submodular-prices} as a direct revelation mechanism, if one also has access to demand oracles for the valuations.


For gross substitutes valuations,
demand queries can be implemented with a polynomial number of value queries \cite{Renato13},
and an optimal allocation can be computed in polynomial time using demand queries \cite{BN07}.
%
The following corollary follows:
\begin{corollary}
\label{cor:gs-prices}
Given sample access to gross substitutes distributions $\dists$ and value oracle access to each valuation in the support of $\dists$,
for every $\eps>0$, one can compute item prices in time $\poly(m,n,1/\eps)$, such that, for any buyer arrival order, the expected welfare of the posted price mechanism is at least $\frac{1}{2}\E_{\vals \sim \dists}[\SW(\alg(\vals))] - \eps$.
\end{corollary}

Here, the mechanism from Corollary~\ref{cor:submodular-prices} can be implemented as a direct revelation mechanism, using only value queries.


\paragraph{Proof of Theorem \ref{th:XOS_Bayesian_Comp}.}
We now proceed with the proof of Theorem \ref{th:XOS_Bayesian_Comp}.  The proof will proceed in two parts.
We begin with Lemma \ref{lem:xos-exists}, which establishes the \emph{existence} of prices that achieve the desired welfare properties, without regard for computation.  In fact, Lemma \ref{lem:xos-exists} will also establish something stronger: if the prices are perturbed slightly, this does not have too large an effect on expected welfare.  We will then use this stronger property to show how the prices can be computed efficiently via sampling.  This sampling process generates the additional additive error term in Theorem~\ref{th:XOS_Bayesian_Comp}.


Before delving into the details of the proof, we need the following definition of an item's welfare contribution.
Fix a valuation profile $\vals=(\vali[1],\dots,\vali[n])$ and algorithm $\alg$, and let $\allocs=(\alloci[1],\dots,\alloci[n])$ be the allocation $\alg(\vals)$.
For each XOS valuation function $\vali(\cdot)$, define the {\em corresponding additive representative} function for the set $\alloci$ as the function $A_i(\cdot)$ satisfying: (i) $\vali(S)\ge A_i(S)$ for any $S\subset[m]$, and (ii) $\vali(\alloci)=A_i(\alloci)$. For every item $j\in\alloci$ we define $\SW_j(\vals):=A_i(\{j\}).$  We think of $\SW_j(\vals)$ as the contribution of item $j$ to the social welfare under valuation profile $\vals$.


\begin{lemma}
\label{lem:xos-exists}
Given a distribution $\dists$ over $\XOS$ valuations, let $\prices$ be the price vector defined as
\[
\pricei[j]=\frac{1}{2}\cdot\Exlong[\vals\sim\dists]{\SW_j(\vals)}.
\]
Let $\pricesp$ be any price vector such that $|\priceip[j] - \pricei[j]| < \delta$ for all $j$.
Then, for any arrival order $\pi$, consumption under prices $\prices'$ results in expected welfare at least $\frac{1}{2}\E_{\vals \sim \dists}[\SW(\alg(\vals))] - m\delta$.
\end{lemma}

\begin{proof}
First, by the definition of $\pricei[j]$,
\begin{align}
\priceip[j]  & = \Exlong[\vals\sim\dists]{\SW_j(\vals)-\priceip[j]} + 2(\priceip[j]-\pricei[j])  \label{eq:pricej-comp}\\
& = \sum_{i=1}^{n}\Exlong[\vals\sim\dists]{\InParentheses{\SW_j(\vals)-\priceip[j]}\cdot\Idr{j\in\alloci(\vals)}} + 2(\priceip[j]-\pricei[j]).\nonumber
\end{align}

We are now going to estimate the sum of buyers' utilities in expectation over $\dists$. Fix $i$ and $\vals=(\vali,\valsmi)$.
Let $\soldi(\vals,\pi)$  denote the set of items that have been sold before the arrival of buyer $i$.
Recall that buyer $i$ picks an allocation\footnote{Note that if a buyer has more than one bundle in his demand correspondence, then we assume that ties can be broken arbitrarily -- even adversarially.}
that maximizes his utility with respect to his valuation $\vali$ and prices $\prices$, from among the items in $M \setminus \soldi(\vals,\pi)$.

Consider another random valuation profile $\walsmi\sim\distsmi$ which is independent of $\vals$.  Let $\alloci(\vali,\walsmi)$ be the allocation returned by $\alg$ on input $(\vali,\walsmi)$. We consider additive representative function $A_i$ for the set $\alloci(\vali,\walsmi)$, so that $A_i(\{j\})=\SW_j(\vali,\walsmi)$ for each $j\in\alloci(\vali,\walsmi)$.
Let $S_i(\vali,\valsmi,\walsmi):=\alloci(\vali,\walsmi)\setminus\soldi(\vals,\pi)$ be the subset of items in $\alloci(\vali,\walsmi)$ that are available to be purchased when buyer $i$ arrives. We note that buyer $i$ could have picked the set $S_i(\vali,\valsmi,\walsmi)$ and, therefore, his utility must be at least the utility he would get from purchasing that set.  Thus we have
\[
\utili(\vals)\ge\Exlong[\walsmi]{\sum_{j\in S_i(\vali,\valsmi,\walsmi)}\max\InParentheses{\SW_j(\vali,\walsmi) -\priceip[j],0}}.
\]
Adding these inequalities for all buyers and taking the expectation over all $\vals\sim\dists$ we get
\begin{align}
\Exlong[\vals\sim\dists]{\sum_{i=1}^{n}\utili(\vals)}  & \ge \sum_{j\in M}\sum_{i=1}^{n}
\operatornamewithlimits{\mathbf E}\limits_{\substack{\vali,\valsmi,\\ \walsmi}}
\left[
\Idr{j\in\alloci(\vali,\walsmi)}\vphantom{E\max_{\walsmi}}
\cdot\max\InParentheses{\SW_j(\vali,\walsmi)- \priceip[j],0}
\cdot\Idr{j\notin\soldi(\vals,\pi)}
\right].
\label{eq:sumutil-comp}
\end{align}
%
We further observe that $\soldi(\vals,\pi)$ does not depend on $\vali$.  That is, $\soldi(\vals,\pi)=\soldi(\valsmi,\pi)$. Therefore, we can rewrite \eqref{eq:sumutil-comp} as follows:
\begin{align}
\Exlong[\vals\sim\dists]{\sum_{i=1}^{n}\utili(\vals)}   
& \ge  \sum_{j\in M}\sum_{i=1}^{n}
\Prlong[\vals]{j\notin\soldi(\vals,\pi)}
\cdot \Exlong[\vali,\walsmi]{\max\InParentheses{\SW_j(\vali,\walsmi)-\priceip[j],0}\cdot
\Idr{j\in\alloci(\vali,\walsmi)}} \nonumber\\
& \ge  \sum_{j\in M}\sum_{i=1}^{n} \Prlong[\vals]{j\notin\sold(\vals,\pi)}
\cdot \Exlong[\vali,\walsmi]{\max\InParentheses{\SW_j(\vali,\walsmi)-\priceip[j],0}\cdot
\Idr{j\in\alloci(\vali,\walsmi)}}\nonumber\\
& \ge  \sum_{j\in M}\Prlong[\vals]{j\notin\sold(\vals,\pi)}
\cdot \left( \sum_{i=1}^{n}\Exlong[\vals]{\InParentheses{\SW_j(\vals)-\priceip[j]}\cdot
\Idr{j\in\alloci(\vals)}}\right) \nonumber\\
& = \sum_{j\in M}\Prlong[\vals]{j\notin\sold(\vals,\pi)}\cdot(\pricei[j] + (\pricei[j]- \priceip[j])).
\label{eq:estimateutil-comp}
\end{align}
%
In the second inequality, we decreased each probability $\Prx{j\notin\soldi(\vals,\pi)}$ to $\Prx{j\notin\sold(\vals,\pi)}$; the inequality holds as all the terms in the summation are non negative. In the third inequality we decreased the random variables under expectations and substituted every variable $(\vali,\walsmi)$ to $\vals$. The last equality follows from \eqref{eq:pricej-comp}.  Inequality \eqref{eq:estimateutil-comp} is our desired bound on the sum of buyer utilities.

We now turn to the expected revenue, which is
\begin{align}
\Exlong[\vals\sim\dists]{\text{Rev}(\vals,\pi)} = \label{eq:revenue-comp} 
\sum_{j\in M}\Prlong[\vals]{j\in\sold(\vals,\pi)}\cdot(\pricei[j]-(\pricei[j]-\priceip[j])).
\end{align}

Therefore, adding \eqref{eq:estimateutil-comp} and \eqref{eq:revenue-comp} we derive the following bound on the expected social welfare:
\begin{align}
\Exlong[\vals\sim\dists]{\sum_{i=1}^{n}\utili(\vals)}+\Exlong[\vals\sim\dists]{\text{Rev}(\vals,\pi)}
 & \ge \sum_{j\in M}\pricei[j] 
+ \sum_{j\in M}(\pricei[j]-\priceip[j])\left(1-2\Prlong[\vals]{j\in\sold(\vals,\pi)}\right)\nonumber\\
& \ge  \frac{1}{2}\Exlong[\vals\sim\dists]{\sum_{i=1}^{n}\vali(\alloci)} - \sum_{j\in M}|\pricei[j]-\priceip[j]|
\nonumber\\
& \ge \frac{1}{2}\Exlong[\vals\sim\dists]{\sum_{i=1}^{n}\vali(\alloci)} - m\delta \nonumber
\end{align}
as required.
\end{proof}

With Lemma \ref{lem:xos-exists} at hand, we are ready to complete the proof of Theorem \ref{th:XOS_Bayesian_Comp}.

\begin{proofsketch} (of Theorem \ref{th:XOS_Bayesian_Comp})
%
%
It remains to show how to compute an appropriate choice of prices $\prices'$ satisfying the conditions of Lemma \ref{lem:xos-exists}.
Our approach will be to estimate $\pricei[j] = \frac{1}{2}\cdot\Exlong[\vals\sim\dists]{\SW_j(\vals)}$ by repeatedly sampling a valuation profile $\hat{\vals} \sim \dists$ and computing $\frac{1}{2}\SW_j(\hat{\vals})$.  Since $\frac{1}{2}\SW_j(\hat{\vals})$ is a random variable lying in $[0,1]$, standard concentration bounds imply that we can accurately estimate its expectation in a relatively small number $t$ of samples.  In Appendix \ref{app:price-sample} we work out the appropriate bounds and show that $t = (\log m + \log n - \log\eps)4m^2 /\eps^2$ samples per item are sufficient to satisfy the conditions of Theorem \ref{th:XOS_Bayesian_Comp}.

\begin{algorithm}[h]
\begin{algorithmic}[1]
\STATE For each item $j \in M$:\\
\STATE \quad Repeat  $t$ times:\\
\STATE \quad \quad Draw $\vals \sim \dists$ and let $\allocs = \alg(\vals)$.\\
\STATE \quad \quad Let $i$ be the agent for which $j \in \alloci$.\\
\STATE \quad \quad Query the $\XOS$ oracle for $\vali$ to find $\SW_j(\vals)$.\\
\STATE \quad Let $\priceip[j]$ be half of the average value of $\SW_j(\vals)$ seen over all $t$ iterations.\\
\STATE return $\prices'$
\end{algorithmic}
\caption{- Price computation algorithm, paramaterized by positive integer $t$.}
\label{alg:price-computation}
\end{algorithm}

We can therefore take $\priceip[j]$ to be the empirical estimate after this number of samples, satisfying the conditions of the theorem.
%
To summarize, this procedure for computing $\pricesp$ is listed formally as Algorithm \ref{alg:price-computation}.
\end{proofsketch}

%% file: MPH.tex

A result similar to Theorem \ref{th:XOS_Bayesian_Comp}
holds for the more general class of
$\MOPH$-$k$ valuations, where we get $O(k)$-approximate DSIC mechanisms for functions with complementarity level $k$. 
We will begin by formally defining the {\em maximum over positive hypergraphs} ($\MOPH$) hierarchy and providing other preliminaries.  We will then provide a formal result statement.

\subsection{Preliminaries and Definitions. }
To explain {\em maximum over positive hypergraphs} ($\MOPH$) hierarchy, we first need a few preliminaries.
A hypergraph representation $h$ of valuation function $v:2^M \to \nonnegR$ is a set function that satisfies $v(S) = \sum_{T \subseteq S} h(T)$. It is easy to verify that any valuation function $v$ admits a unique hypergraph representation and vice versa.
A set $S$ such that $h(S) \neq 0$ is said to be a {\em hyperedge} of $h$.
The hypergraph representation can be thought as a weighted hypergraph, where every vertex is associated with an item in $M$, and the weight of each hyperedge $e\subseteq M$ is $h(e)$. Then the value of the function for any set $S\subseteq M$ is the total value of all hyperedges that are contained in $S$.

The {\em rank} of a hypergraph representation $h$ is the cardinality $k$ of the largest hyperedge.
The rank of $v$ is the rank of its corresponding $h$ and we refer to a valuation function $v$ with rank $k$ as a \emph{hypergraph-$k$} valuation. If the hypergraph representation of $v$ is non-negative, i.e. for any $S\subseteq M$, $h(S)\geq 0$, then we refer to function $v$ as a \emph{positive hypergraph-$k$} function (PH-$k$) \cite{Abraham12}.
We are now ready to present the class of $\MOPH$-$k$ valuations.

\begin{definition}[$\MOPH$-$k$ valuation]
A monotone valuation function $v:2^M\to \nonnegR$ is {\em Maximum over Positive Hypergraph-$k$} ($\MOPH$-$k$) if it can be expressed as a maximum over a set of PH-$k$ functions.  That is, there exist PH-$k$ functions $\{v_{\ell}\}_{\ell\in\Ell}$ such that for every set $S \subseteq M$,
\begin{equation}
\textstyle{v(S) = \max_{\ell \in \Ell} v_{\ell}(S)},
\end{equation}
where $\Ell$ is an arbitrary index set.
\end{definition}

It can be easily verified that the highest level of the hierarchy, $\MOPH$-$m$ captures all monotone functions, and the lowest level, $\MOPH$-$1$,
captures all $\XOS$ functions.


Finally, we define what is meant by an \MPHk oracle, which is an extension of XOS oracles to higher levels of the $\MOPH$ hierarchy.
Suppose that valuation function $v$ is \MPHk, with supporting PH-$k$ functions $\{v_{\ell}\}_{\ell\in\Ell}$.  An \MPHk-oracle for $v$ takes as input a set of items $S$, and returns the PH-$k$ function $v_{\ell}$ for which $v(S) = v_{\ell}(S)$.  We will assume that this function $v_{\ell}$ is returned in its explicit hypergraph representation, i.e.\ as a list of weighted hyperedges.  Note that the size of this representation depends on the number of hyperedges required to express the PH-$k$ functions $v_{\ell}$, and is at most $O(m^k)$.  On a side note, it is this bound that leads to a runtime that is polynomial in $m^k$ in Theorem \ref{th:mohk_Bayesian_Comp}.  Note that if we restricted attention to \MPHk valuations whose supporting PH-$k$ functions each have at most $r$ hyperedges, then this runtime dependency would change from $m^k$ to $r$.

\subsection{Pricing for $\MOPH$-$k$ valuations. }
Our result is cast in the following theorem, whose proof is deferred to Appendix \ref{app:mph}.\footnote{The proof of Theorem~\ref{th:mohk_Bayesian_Comp} follows Theorem~\ref{th:XOS_Bayesian_Comp}, but with an important difference: the accounting of the contribution of an item $j$ to the welfare is more complex, since one must consider all hyperedges in which $j$ appears.  This complicates the choice of prices, as well as the derivation of welfare bounds.
}

\begin{theorem}
\label{th:mohk_Bayesian_Comp}
Suppose our Bayesian instance $\dists$ over $\MOPH$ valuations is given via a sample access to $\dists$.
Suppose that for every $\vals\sim\dists$ we have
\begin{enumerate}
\item black-box access to a welfare maximization algorithm $\alg$ for combinatorial auctions,
\item an $\MOPH$ query oracle for the valuations in the support of $\dists$.
\end{enumerate}
Then, for every $\eps>0$, one can compute item prices in time\footnote{The exponential dependence on $k$ in the runtime is related to the representation of $\MOPH$ valuations.  In particular, the output of an $\MOPH$-$k$ oracle can be of size $O(m^k)$.  One could reduce this bound by imposing constraints on the complexity of a valuation's $\MOPH$ representation.  This is discussed further in Appendix \ref{app:mph}.} $\poly(m^k,n,1/\eps)$ that generate expected welfare of at least $\frac{1}{4k}\E_{\vals \sim \dists}[\SW(\alg(\vals))] - \eps$ for any buyers' arrival order.
\end{theorem}

We also show that this result is essentially tight. Indeed, for each level $k$ of  complementarities across the items in the $\MOPH$-$k$ hierarchy, we may consider single minded (of size $k$) and unit-demand valuations. It turns out that item prices may result in an outcome with a linear (in the number of items) loss in social welfare. The example is deferred to Appendix~\ref{app:mph-lower}. 

%% file: conclusions.tex

We conclude with a few remarks.

First, in our mechanisms, we consider an arbitrary order of arrivals, which may be chosen by an adversary after the prices are posted, but before the adversary observes the realization of the buyer valuations.
It is not difficult to verify that the same results extend to an {\em adaptive} adversary, who chooses the arrival order sequentially; i.e., an adversary who observes which items have been purchased by previous buyers and even the realization of previous buyers' valuations, and chooses the next buyer to arrive based on this information. Our proof techniques (in Theorems~\ref{th:XOS_Bayesian_Comp} and \ref{th:mohk_Bayesian_Comp}) apply to this adaptive adversary as well.

Second, readers who are familiar with literature on Walrasian equilibrium will realize the similarities between the two models, but also the stark contrast.
The main difference is whether agents arrive to the market {\em sequentially} (as in our model), or {\em simultaneously} (as in a Walrasian equilibrium).
Recent results~\cite{FGL13} have shown that in the simultaneous model (even when some items may remain unsold), there may be a {\em linear} loss in welfare for XOS buyers, even in a full information setting.
Thus our work demonstrates a strong gap in welfare between simultaneous and sequential arrivals, when restricted to individual demand satisfaction.


Our model and results leave a number of directions for future research.
First, the constant approximation for XOS valuations implies (by known results, see e.g.~\cite{BR11}) a logarithmic approximation for subadditive valuations.
It remains open whether a constant approximation for subadditive valuations can be achieved.

In this work we focused on a setting in which the buyer arrival order is adversarial, but one might consider relaxing this worst-case setting.  For example: does the approximability of the problem substantially improve if the buyers arrive in a uniformly random order?  What if the mechanism can select the order, subject to incentive compatibility constraints?  Alternatively, one might ask whether a constant approximation is still possible if the adversary is more powerful, and can observe all valuation realizations before selecting the arrival order.  Also, rather than looking at an exogenously-imposed arrival order, one might imagine bidders strategically choosing their arrival times, with early positions in queue being costly to secure.

Finally, throughout the paper we assume that items are indivisible and heterogeneous. It would be interesting to partially relax these assumptions. For example, one could assume that every item in the market has a few identical copies and that every buyer wants at most a single copy of each item. It would be interesting to analyze the efficiency of posted price mechanisms as a function of the minimal number of item copies. Given the negative results for valuations with high degree of complementarity, it would be particularly interesting to find relaxations that admit positive results, say for single-minded buyers.

%% file: price-sample-appendix.tex
We now present the details omitted from the proof of Theorem \ref{th:XOS_Bayesian_Comp}.

We wish to show how to compute an appropriate choice of prices $\prices'$ satisfying the conditions of Lemma \ref{lem:xos-exists}.
Our approach will be to estimate $\pricei[j] = \frac{1}{2}\cdot\Exlong[\vals\sim\dists]{\SW_j(\vals)}$ by repeatedly sampling a valuation profile $\hat{\vals} \sim \dists$ and computing $\frac{1}{2}\SW_j(\hat{\vals})$.  Since $\frac{1}{2}\SW_j(\hat{\vals})$ is a random variable lying in $[0,1]$, standard concentration bounds imply that we can accurately estimate its expectation in a relatively small number $t$ of samples.  We can therefore take $\priceip[j]$ to be the empirical estimate after this number of samples, satisfying the conditions of the theorem.  This procedure is listed formally as Algorithm \ref{alg:price-computation} in Section \ref{sec:xos}.

%

We wish to choose $t$ large enough that, with probability at least $1 - \eps/n$, we will have $|\priceip[j] - \pricei[j]| < \eps/2m$ for all $j$.
Fix any $j$ and note that $\priceip[j]$ is the average of $t$ identical samples from a distribution supported on $[0,1]$, with expected value $\pricei[j]$.
Thus, by the Hoeffding bound, we have that
\[ \Pr[ |\priceip[j] - \pricei[j]| > \eps/2m ] < 2e^{- t(\eps/2m)^2}. \]
We can therefore choose $t = (\log m + \log n - \log\eps)4m^2 /\eps^2$ to get $\Pr[ |\priceip[j] - \pricei[j]| > \eps/2m ] < \eps/mn$.
Applying a union bound over all $j \in M$, we have that $|\priceip[j] - \pricei[j]| < \eps/2m$ for all $j$
with probability at least $1 - \eps/n$, as desired.

Setting $\delta = \eps/2m$ in Lemma~\ref{lem:xos-exists}, we have that our computed prices generate welfare at least $\frac{1}{2}\E_{\vals \sim \dists}[\SW(\alg(\vals))] - \eps/2$, with probability at least $1 - \eps/n$.
We conclude that our computed prices generate an expected welfare of at least
\begin{align*}
\left( \frac{1}{2}\E_{\vals \sim \dists}[\SW(\alg(\vals))] - \frac{\eps}{2} \right) \left( 1-\frac{\eps}{n} \right) > 
\frac{1}{2}\E_{\vals \sim \dists}[\SW(\alg(\vals))] - \eps, 
\end{align*}
as required.
The last inequality follows since $\E_{\vals \sim \dists}[\SW(\alg(\vals))]\leq\E_{\vals \sim \dists}[\sum_{i=1}^{n}\vali(M)] \leq n$.

%% file: MPH-app.tex
\noindent
We closely follow the proof of Theorem~\ref{th:XOS_Bayesian_Comp} for XOS buyers. However, there is an extra difficulty for $\MOPH$-$k$ valuations, since  the concept of the ``contribution of an item to welfare'' is not as straightforward as for $XOS$ valuations. Our main new challenge will be to appropriately account for the contributions of different items. 

We first describe an ideal price vector $\prices$ which we would like to use for the distribution $\dists$. For each fixed valuation profile $\vals=(\vali[1],\dots,\vali[n])$ we consider allocation $\allocs(\vals)=(\alloci[1](\vals),\dots,\alloci[n](\vals))$ returned by black-box algorithm $\alg$. For each \MPHk valuation function $\vali(\cdot)$ we take the respective hypergraph representative function $A_i(\cdot)$ for the set $\alloci(\vals)$, i.e., $\vali(S)\ge A_i(S)$ for any $S\subset[m]$ and $\vali(\alloci(\vals))=A_i(\alloci(\vals))$.  Write $\weight_i(\cdot)$ for the hypergraph weights corresponding to the hypergraph function $A_i(\cdot)$; then, by definition, $A_i(S) = \sum_{T \subseteq S}\weight_i(T)$ for all $S \subseteq \alloci(\vals)$.

For every $\vals$, every buyer $i$, and each item $j\in\alloci(\vals)$, we define
\[
\pricei[j](\vals) = \frac{1}{\alpha}\sum_{\substack{T \ni j\\T \subseteq \alloci(\vals)}}\frac{\weight_i(T)}{|T|},
\]
where $\alpha$ is a constant to be determined later.  The price vector $\prices(\vals)$ has a natural interpretation: for each hyperedge in the hypergraph function $A_i(\alloci(\vals))$, divide its weight uniformly among the items in that edge; the price of item $j$ is then the total weight allocated to item $j$, scaled down by factor $\alpha$.  The price $\pricei[j](\vals)$ for item $j$ is the ideal price we would like to set in the full-information setting, if we knew the valuation profile $\vals$.

We can now define an ideal price of item $j$ in the Bayesian setting, which will be
\[
\pricei[j]=\Exlong[\vals\sim\dists]{\pricei[j](\vals)}.
\]

The following Lemma relates the full-information prices for a subset of items to the marginal impact on a buyer's value if those items are removed from an allocation.
\begin{lemma}
\label{lem:mohk_price}
For any $\vals$, any buyer $i$, and any $Q \subseteq \alloci(\vals)$,
\[ \vali( \alloci(\vals) \backslash Q ) + \alpha k \cdot \sum_{j \in Q} \pricei[j](\vals) \ge \vali(\alloci(\vals)). \]
\end{lemma}
\begin{proof}
\begin{align*}
\vali( \alloci(\vals) \backslash Q ) + \alpha k \cdot \sum_{j \in Q} \pricei[j](\vals)
& = \sum_{T \subseteq \alloci(\vals) \backslash Q} \weight_i(T) 
+ \alpha k \cdot \sum_{j \in Q} \frac{1}{\alpha} \cdot \sum_{\substack{T \ni j\\T \subseteq \alloci(\vals)}} \frac{\weight_i(T)}{|T|} \\
&\ge \sum_{T \subseteq \alloci(\vals) \backslash Q} \weight_i(T) 
+ \sum_{j \in Q} \sum_{\substack{T \ni j\\T \subseteq \alloci(\vals)}} \weight_i(T) \\
&\ge \sum_{T \subseteq \alloci(\vals) \backslash Q} \weight_i(T) 
+ \sum_{\substack{T \subseteq \alloci(\vals)\\T \cap Q \neq \emptyset}} \weight_i(T) \\
&= \sum_{T \subseteq \alloci(\vals)} \weight_i(T) \\
&= \vali( \alloci(\vals) )
\end{align*}
where the first inequality follows because $\weight_i(T) > 0$ only for $T$ with $|T| \leq k$, and the second inequality follows by noting that each hyperedge $T$ counted in the second summation must have a non-empty intersection with $Q$ and is counted $|T \cap Q| \geq 1$ times.
\end{proof}


The next Lemma estimates the expected social welfare of a mechanism with posted prices that are close to the ideal $\prices$.

\begin{lemma}
\label{lem:mohk_exists}
Let $\pricesp$ be such that $|\priceip[j] - \pricei[j]| < \delta$ for all $j$.
Then consumption under prices $\prices'$ results in expected welfare of at least $\frac{1}{4k}\E_{\vals \sim \dists}[\SW(\alg(\vals))] - 2m\delta$.
\end{lemma}

\begin{proof} Given the prices $\pricesp$, let $\pi$ be the (adversarial) order of arrival. We are going to bound the sum of buyers' utilities in expectation over $\dists$. To do so, for each fixed $i$ and $\vals=(\vali,\valsmi)$, we consider another random valuation profile $\walsmi\sim\distsmi$, drawn independently of $\vals$.  Consider also $\alloci(\vali,\walsmi)$, the allocation returned by $\alg$ on the valuation profile $(\vali,\walsmi)$.
Let $S_i(\vali,\valsmi,\walsmi):=\alloci(\vali,\walsmi)\cap\soldi(\vals,\pi)$ be the subset of items in $\alloci(\vali,\walsmi)$ that are already sold when buyer $i$ is selected to make a purchase.  Let $R_i(\vali,\valsmi,\walsmi):=\alloci(\vali,\walsmi)\setminus\soldi(\vals,\pi)$ be the subset of items in $\alloci(\vali,\walsmi)$ that remain unsold at this time.  We note that buyer $i$ could have picked the set $R_i(\vali,\valsmi,\walsmi)$ and, therefore, his utility is at least the utility he would get from this set. Thus we have

\begin{align*}
\utili(\vals) &\ge \Exlong[\walsmi]{\vali(R_i(\vali,\valsmi,\walsmi)) - \sum_{j \in R_i(\vali,\valsmi,\walsmi)}\priceip[j]}\\
& \ge \Exlong[\walsmi]{\vali(R_i(\vali,\valsmi,\walsmi)) - \sum_{j \in X_i(\vali,\walsmi)}\priceip[j]}.
\end{align*}

Applying Lemma \ref{lem:mohk_price} to valuation profile $(\vali,\walsmi)$ and set $Q = \soldi(\vals,\pi)$, we conclude

\begin{align*}
\utili(\vals) & \ge \operatornamewithlimits{\mathbf E}\limits_{\walsmi}\left[
 \alpha\cdot k\cdot \sum_{j \in S_i(\vali,\valsmi,\walsmi)}\pricei[j](\vali,\walsmi) 
- \sum_{j \in X_i(\vali,\walsmi)}\priceip[j]
\right] .
\end{align*}

We now sum over all $i$ and take an expectation over $\vals \sim \dists$ to conclude that

\begin{align}
\label{eq:mohk_1}
\Exlong[\vals]{\sum_i \utili(\vals)} 
& \ge  \sum_i \Exlong[\vals,\walsmi]{\vali(\alloci(\vali,\walsmi)) - \sum_{j \in X_i(\vali,\walsmi)}\priceip[j]} \nonumber\\
& \quad\quad 
- \alpha k \cdot \sum_i \Exlong[\vals,\walsmi]{\sum_{j \in S_i(\vali,\valsmi,\walsmi)}\pricei[j](\vali,\walsmi)}.
\end{align}

Let us analyze separately the two summations on the RHS of \eqref{eq:mohk_1}.  For the first summation, note that $\valsmi$ does not appear in the expression within the expectation.  Thus, by applying a change of variables and then using linearity of expectation, we have
\begin{align}
\sum_i \Exlong[\vals,\walsmi]{\vali(\alloci(\vali,\walsmi)) - \sum_{j \in X_i(\vali,\walsmi)}\priceip[j]} 
& = \sum_i \Exlong[\vals]{\vali(\alloci(\vals)) - \sum_{j \in X_i(\vals)}\priceip[j]} \nonumber \\
& = \Exlong[\vals]{\sum_i \vali(\alloci(\vals))} - \sum_j \priceip[j] \nonumber \\
& \ge \Exlong[\vals]{\sum_i \vali(\alloci(\vals))} - \sum_j \pricei[j] - \delta m \nonumber\\
& = \Exlong[\vals]{\sum_i \vali(\alloci(\vals))} 
- \frac{1}{\alpha}\cdot\sum_j \Exlong[\vals]{\sum_i\sum_{\substack{T\ni j\\T \subseteq \alloci(\vals)}} \frac{\weight_i(T)}{|T|}} \nonumber - \delta m\\
& = \Exlong[\vals]{\sum_i \vali(\alloci(\vals))} - \frac{1}{\alpha} \cdot \Exlong[\vals]{\sum_i \vali(\alloci(\vals))} - \delta m \nonumber \\
& = \left(1 - \frac{1}{\alpha}\right)\Exlong[\vals]{\sum_i \vali(\alloci(\vals))} - \delta m. \label{eq:mohk_2}
\end{align}
Note that the inequality follows from the fact that $|\pricei[j] - \priceip[j]| < \delta$ for each item $j$.

For the second summation on the RHS of \eqref{eq:mohk_1}, we first recall that sets $\soldi$ and $S_i(\vali,\valsmi,\walsmi)$ are defined for the prices $\pricesp$.
Further note that since $\valsmi$ and $\walsmi$ are drawn independently, we have
\begin{align}
\label{eq:mohk_3}
\sum_i \Exlong[\vals,\walsmi]{\sum_{j \in S_i(\vals,\walsmi)}\pricei[j](\vali,\walsmi)} 
& = \sum_i \operatornamewithlimits{\mathbf E}\limits_{\vals,\walsmi}\left[\sum_{j} \pricei[j](\vali,\walsmi) 
\cdot\Idr{j \in \soldi(\valsmi,\pi)} \cdot \Idr{j \in \alloci(\vali,\walsmi)}\right] \nonumber\\
&  = \sum_{i,j} \Exlong[\vali,\walsmi]{\Idr{j \in \alloci(\vali,\walsmi)} \cdot \pricei[j](\vali,\walsmi)} 
\cdot\Prlong[\valsmi]{j \in \soldi(\valsmi,\pi)} \nonumber\\
&  \le \sum_{i,j} \Exlong[\vali,\walsmi]{\Idr{j \in \alloci(\vali,\walsmi)} \cdot \pricei[j](\vali,\walsmi)} 
\cdot\Prlong[\vals]{j \in \sold(\vals,\pi)}\nonumber\\
&  = \sum_j \Prlong[\vals]{j \in \sold(\vals,\pi)} 
\cdot\sum_i \Exlong[\vali,\walsmi]{\Idr{j \in \alloci(\vali,\walsmi)} \cdot \pricei[j](\vali,\walsmi)} \nonumber\\ 
&  = \sum_j \Prlong[\vals]{j \in \sold(\vals,\pi)} \cdot \pricei[j]\nonumber \\
& \le \sum_j \Prlong[\vals]{j \in \sold(\vals,\pi)} \cdot \priceip[j] + \delta m\nonumber\\
&= \Exlong[\vals\sim\dists]{\text{Rev}(\vals,\pi)}+\delta m. 
\end{align}
The first inequality follows because the probability that $j$ is sold before agent $i$ arrives is dominated by the probability that $j$ is sold at all, and the second inequality follows from the fact that $|\pricei[j] - \priceip[j]| < \delta$ for each item $j$. 


Substituting \eqref{eq:mohk_2} and \eqref{eq:mohk_3} into \eqref{eq:mohk_1}, we have

\begin{align}
\label{eq.mohk.4}
\Exlong[\vals]{\sum_i \utili(\vals)} & \ge \left(1 - \frac{1}{\alpha}\right)\Exlong[\vals]{\sum_i \vali(\alloci(\vals))}
-\delta m - \alpha k \cdot \Exlong[\vals\sim\dists]{\text{Rev}(\vals,\pi)} - \alpha k \delta m.
\end{align}

As long as $\alpha k \geq 1$, we can rearrange and conclude
\begin{align*}
\alpha k \left( \Exlong[\vals\sim\dists]{\sum_{i=1}^{n}\utili(\vals)}+ \Exlong[\vals\sim\dists]{\text{Rev}(\vals,\pi)} \right) 
\ge \left( 1 - \frac{1}{\alpha} \right) \Exlong[\vals\sim\dists]{\sum_{i=1}^{n}\vali(\alloci)}- 2\alpha k \delta m.
\end{align*}
Taking $\alpha = 2$, we conclude that the expected welfare of the Posted Pricing Mechanism is within a factor $4k$ of the expected welfare of $\alg$ and small additive error of $2m\delta$, as required.
\end{proof}

We continue with the proof of Theorem \ref{th:mohk_Bayesian_Comp}. Following the same analysis as in Theorem~\ref{th:XOS_Bayesian_Comp} for each item $j$ we can estimate the price $\priceip[j]$ by sampling $t = (\log m + \log n - \log\eps)16m^2 /\eps^2$ valuation profiles, so that $ \Pr[ |\priceip[j] - \pricei[j]| > \eps/4m ] < \eps/mn$. We compute $\pricei[j]'$ for each sample (using algorithm $\alg$ and the \MPHk query oracle) and take the average of all prices seen. Applying a union bound over all $j \in M$ we obtain a guarantee that $|\priceip[j] - \pricei[j]| < \eps/4m$ for all $j$ with probability at least $1-\eps/n$. Now, by setting $\delta = \eps/4m$ in Lemma~\ref{lem:mohk_exists} we have our computed prices $\pricesp$ to generate welfare of at least $\frac{1}{4k}\E_{\vals \sim \dists}[\SW(\alg(\vals))] - \frac{\eps}{2}$ with probability at least $1-\eps/n$.

Finally, we conclude that generated expected welfare is at least
\begin{align*}
\left( \frac{1}{4k}\E_{\vals \sim \dists}[\SW(\alg(\vals))] - \frac{\eps}{2} \right) \left( 1-\frac{\eps}{n} \right) 
> \frac{1}{4k}\E_{\vals \sim \dists}[\SW(\alg(\vals))] - \eps, 
\end{align*}
as required. The last inequality follows, since $\E_{\vals \sim \dists}[\SW(\alg(\vals))]\le \E_{\vals \sim \dists}[\sum_{i=1}^{n}\vali(M)] \leq n$.

%% file: mph-lower.tex
\paragraph{Example.}
Suppose there are $m$ identical items in the market and two buyers. Let the first buyer have unit-demand valuation $1$ per item
and the second single-minded buyer have value $m-1$ for the set of all $m$ items and $0$ value for any smaller subset. The optimal
social welfare $\opt$ is $m-1$, where the second buyer is allocated all $m$ items.

Let the seller fix prices on the items. We let the first buyer arrive first. He will buy the
cheapest item, if its price is below $1$. Then the second buyer has $0$ value for the remaining items, which results in a social welfare of $1$.
In the case where each item costs at least $1$, the first buyer purchases nothing but so does the second buyer, as he derives value $m-1$ from the entire set, for a total cost of at least $m$. Therefore, the social welfare in the latter case is $0$. 
We conclude that the social welfare does not exceed $1$ in either of the cases, which gives us the claimed linear gap of $m-1$ with respect to the optimal social welfare.

%% file: item-prices.bbl
\begin{thebibliography}{10}

\bibitem{Abraham12}
Ittai Abraham, Moshe Babaioff, Shaddin Dughmi, and Tim Roughgarden.
\newblock Combinatorial auctions with restricted complements.
\newblock In {\em ACM Conference on Electronic Commerce}, pages 3--16, 2012.

\bibitem{Saeed11}
Saeed Alaei.
\newblock Bayesian combinatorial auctions: Expanding single buyer mechanisms to
  many buyers.
\newblock In {\em FOCS}, pages 512--521, 2011.

\bibitem{Babaioff11}
Moshe Babaioff, Liad Blumrosen, Shaddin Dughmi, and Yaron Singer.
\newblock Posting prices with unknown distributions.
\newblock In {\em ICS}, pages 166--178, 2011.

\bibitem{BH-11}
X.~Bei and Z.~Huang.
\newblock Bayesian incentive compatibility via fractional assignments.
\newblock In {\em SODA}, 2011.

\bibitem{BR11}
Kshipra Bhawalkar and Tim Roughgarden.
\newblock Welfare guarantees for combinatorial auctions with item bidding.
\newblock In {\em SODA}, pages 700--709, 2011.

\bibitem{Bhawalkar12}
Kshipra Bhawalkar and Tim Roughgarden.
\newblock Simultaneous single-item auctions.
\newblock In {\em WINE}, 2012.

\bibitem{BN07}
L.~Blumrosen and N.~Nisan.
\newblock Combinatorial auctions.
\newblock In N.~Nisan, T.~Roughgarden, {\'E}.~Tardos, and V.~Vazirani, editors,
  {\em Algorithmic Game Theory}, chapter~11, pages 267--299. Cambridge
  University Press, 2007.

\bibitem{BlumrosenH08}
Liad Blumrosen and Thomas Holenstein.
\newblock Posted prices vs. negotiations: an asymptotic analysis.
\newblock In {\em EC}, page~49, 2008.

\bibitem{Blumrosen2009}
Liad Blumrosen and Noam Nisan.
\newblock On the computational power of demand queries.
\newblock {\em SIAM J. Comput.}, 39(4):1372--1391, 2009.

\bibitem{Chawla07}
Shuchi Chawla, Jason~D. Hartline, and Robert~D. Kleinberg.
\newblock Algorithmic pricing via virtual valuations.
\newblock In {\em ACM Conference on Electronic Commerce}, pages 243--251, 2007.

\bibitem{Chawla10}
Shuchi Chawla, Jason~D. Hartline, David~L. Malec, and Balasubramanian Sivan.
\newblock Multi-parameter mechanism design and sequential posted pricing.
\newblock In {\em STOC}, pages 311--320, 2010.

\bibitem{ChawlaMS10}
Shuchi Chawla, David~L. Malec, and Balasubramanian Sivan.
\newblock {The Power of Randomness in Bayesian Optimal Mechanism Design}.
\newblock In {\em the 11th ACM Conference on Electronic Commerce (EC)}, 2010.

\bibitem{CKS08}
George Christodoulou, Annam{\'a}ria Kov{\'a}cs, and Michael Schapira.
\newblock Bayesian combinatorial auctions.
\newblock In {\em ICALP}, pages 820--832, 2008.

\bibitem{DhangwatnotaiRY10}
Peerapong Dhangwatnotai, Tim Roughgarden, and Qiqi Yan.
\newblock Revenue maximization with a single sample.
\newblock In {\em EC}, pages 129--138, 2010.

\bibitem{Dobzinski07}
Shahar Dobzinski.
\newblock Two randomized mechanisms for combinatorial auctions.
\newblock In {\em APPROX-RANDOM}, pages 89--103, 2007.

\bibitem{DobzinskiV12}
Shahar Dobzinski and Jan Vondr{\'a}k.
\newblock The computational complexity of truthfulness in combinatorial
  auctions.
\newblock In {\em EC}, pages 405--422, 2012.

\bibitem{DughmiV11}
Shaddin Dughmi and Jan Vondr{\'a}k.
\newblock Limitations of randomized mechanisms for combinatorial auctions.
\newblock In {\em FOCS}, pages 502--511, 2011.

\bibitem{Edelman05}
Benjamin Edelman, Michael Ostrovsky, and Michael Schwarz.
\newblock Internet advertising and the generalized second price auction:
  Selling billions of dollars worth of keywords.
\newblock Technical report, National Bureau of Economic Research, 2005.

\bibitem{Feige09}
Uriel Feige.
\newblock On maximizing welfare when utility functions are subadditive.
\newblock {\em SIAM J. Comput.}, 39(1):122--142, 2009.

\bibitem{Feige14}
Uriel Feige, Michal Feldman, Nicole Immorlica, Rani Izsak, Brendan Lucier, and
  Vasilis Syrgkanis.
\newblock A unified approach to valuations with restricted complements.
\newblock In {\em Working paper}, 2014.

\bibitem{FFGL13}
Michal Feldman, Hu~Fu, Nick Gravin, and Brendan Lucier.
\newblock Simultaneous auctions are (almost) efficient.
\newblock In {\em STOC}, pages 201--210, 2013.

\bibitem{FGL13}
Michal Feldman, Nick Gravin, and Brendan Lucier.
\newblock Combinatorial walrasian equilibrium.
\newblock In {\em STOC}, pages 61--70, 2013.

\bibitem{Gul1999}
Faruk Gul and Ennio Stacchetti.
\newblock Walrasian equilibrium with gross substitutes.
\newblock {\em J. of Economic Theory}, 87(1):95--124, 1999.

\bibitem{HartlineKM11}
Jason~D. Hartline, Robert Kleinberg, and Azarakhsh Malekian.
\newblock Bayesian incentive compatibility via matchings.
\newblock In {\em SODA}, pages 734--747, 2011.

\bibitem{HartlineL10}
Jason~D. Hartline and Brendan Lucier.
\newblock Bayesian algorithmic mechanism design.
\newblock In {\em STOC}, 2010.

\bibitem{HartlineR09}
Jason~D. Hartline and Tim Roughgarden.
\newblock Simple versus optimal mechanisms.
\newblock In {\em EC}, pages 225--234, 2009.

\bibitem{HKMN11}
Avinatan Hassidim, Haim Kaplan, Yishay Mansour, and Noam Nisan.
\newblock Non-price equilibria in markets of discrete goods.
\newblock In {\em EC}, pages 295--296, 2011.

\bibitem{HK92}
Theodore~P. Hill and Robert~P. Kertz.
\newblock A survey of prophet inequalities in optimal stopping theory.
\newblock {\em CONTEMPORARY MATHEMATICS}, 125, 1992.

\bibitem{Kelso1982}
Jr~Kelso, Alexander~S and Vincent~P Crawford.
\newblock Job matching, coalition formation, and gross substitutes.
\newblock {\em Econometrica}, 50(6):1483--1504, November 1982.

\bibitem{KW12}
Robert Kleinberg and S.~Matthew Weinberg.
\newblock Matroid prophet inequalities.
\newblock In {\em Proceedings of the 44th Symposium on Theory of Computing
  Conference, {STOC} 2012, New York, NY, USA, May 19 - 22, 2012}, pages
  123--136, 2012.

\bibitem{Lehmann2001}
Benny Lehmann, Daniel Lehmann, and Noam Nisan.
\newblock Combinatorial auctions with decreasing marginal utilities.
\newblock {\em Games and Economic Behavior}, 55(2):270 -- 296, 2006.
\newblock Mini Special Issue: Electronic Market Design</ce:title>.

\bibitem{Azevedo2013}
Azevedo~Eduardo M., Glen~E. Weyl, and White Alexander.
\newblock Walrasian equilibrium in large, quasi-linear markets.
\newblock {\em Theoretical Economics}, 8:281--290, 2013.

\bibitem{Milgrom89}
Paul Milgrom.
\newblock {Auctions and Bidding: A Primer}.
\newblock {\em Journal of Economic Perspectives}, 3(3):3--22, 1989.

\bibitem{Milgrom98}
Paul Milgrom.
\newblock Game theory and the spectrum auctions.
\newblock {\em European Economic Review}, 42(3):771--778, 1998.

\bibitem{Renato13}
Renato Paes~Leme.
\newblock Gross substitutability: An algorithmic survey.
\newblock 2013.

\bibitem{Varian07}
Hal~R. Varian.
\newblock Search engine economics.
\newblock In {\em Wirtschaftsinformatik (1)}, pages 29--30, 2007.

\bibitem{Vondrak08}
Jan Vondr{\'a}k.
\newblock Optimal approximation for the submodular welfare problem in the value
  oracle model.
\newblock In {\em STOC}, pages 67--74, 2008.

\end{thebibliography}
